\documentclass[reqno]{amsart}
\usepackage{amssymb}

\newtheorem{theorem}{Theorem}[section]
\newtheorem{lemma}[theorem]{Lemma}
\newtheorem{proposition}[theorem]{Proposition}
\newtheorem{definition}[theorem]{Definition}
\newtheorem{problem}{Problem}
\newtheorem{remark}[theorem]{Remark}
\newtheorem{condition}[theorem]{Hypothesis}
\numberwithin{equation}{section}

\renewcommand{\Im}{\operatorname{Im}}
\renewcommand{\Re}{\operatorname{Re}}
\input{tcilatex}

\begin{document}
\title[Meromorphic solutions]{Spatial Analyticity of solutions to integrable
systems. I. The KdV case}
\author{Alexei Rybkin}
\address{University of Alaska Fairbanks}
\date{September, 2011}
\address{Department of Mathematics and Statistics \\
University of Alaska Fairbanks\\
PO Box 756660\\
Fairbanks, AK 99775}
\email{arybkin@alaska.edu}
\thanks{Based on research supported in part by the NSF under grant DMS
1009673.}
\subjclass[2010]{37K15, 47B35, 35B65}
\keywords{Korteweg-de Vries equation, inverse scattering transform, Schr\"{o}%
dinger operator, Hankel operator, Gevrey regularity. }

\begin{abstract}
We are concerned with the Cauchy problem for the KdV equation for nonsmooth
locally integrable initial profiles $q$'s which are, in a certain sense,
essentially bounded from below and $q\left( x\right) =O\left(
e^{-cx^{\varepsilon }}\right) ,x\rightarrow +\infty $, with some positive $c$
and $\varepsilon $. Using the inverse scattering transform, we show that the
KdV flow turns such initial data into a function which is (1) meromorphic
(in the space variable) on the whole complex plane if $\varepsilon >1/2$,
(2) meromorphic on a strip around the real line if $\varepsilon =1/2$, and
(3) Gevrey regular if $\varepsilon <1/2$. Note that $q$'s need not have any
decay or pattern of behavior at $-\infty $.
\end{abstract}

\maketitle



\section{Introduction and statements of main results}

The gain and persistence of regularity effects are important features of
many dispersive (linear and nonlinear) partial differential equations
(PDEs). The literature on the subject is truly enormous and we make no
attempt to give a comprehensive review here. We only mention two recent
relevant papers by Himonas et al \cite{Han_Himonas_2011}, \cite{Himonas_2011}
where the interested reader can find further references on analytic and
Gevrey regularity properties for KdV-type equations. In fact, we are
interested in a much stronger effect of formation of meromorphic solution
out of nonsmooth data. More specifically, in the current paper, we are
concerned with the following problem.


\begin{problem}
\label{pb1} Given the Cauchy problem for the KdV equation\footnote{%
We use $z$ instead of $x$ for the spatial variable as it will frequently be
complex.} 
\begin{equation}
\begin{cases}
\partial _{t}u-6u\partial _{z}u+\partial _{z}^{3}u=0 \\ 
u|_{t=0}=q%
\end{cases}%
,  \label{eq1.1}
\end{equation}%
describe the largest possible class of (non-smooth) initial data $q$ which
evolve into functions $u(z,t)$ meromorphic with respect to $z$ for any $t>0$.
\end{problem}

Meromorphic (or, more generally, analytic) solutions have of course been
intensively studied since the boom around integrable systems started in the
late 60s. A\ pure soliton (reflectionless) solution, historically the first
explicit solution, is meromorphic on the whole complex plane having
infinitely many double poles. This fact is of course a trivial observation
immediately following from the explicit formula for multisoliton solutions.
We emphasize that how those poles interact is not obvious at all. This
question was raised back in earlier 70s by Kruskal and has been followed up
by many. We refer the interested reader to a particularly influential 1977
paper \cite{Air77} by Airault-McKean-Moser and recent Bona-Weissler \cite%
{Bona09} and the literature cited therein. More complicated examples of
explicit solutions include algebraric, rational, meromorphic simply
periodic, elliptic, etc. (see, e.g. \cite{AblSatJMP78}, \cite{Birnir87}, 
\cite{GUW06}, \cite{AktMee06} and the literature cited therein). All these
examples are of course very specific and in addition those $q$'s are already
meromorphic (i.e. smooth on the real line). Although Problem \ref{pb1} is
not addressed in those papers but they demonstrate the importance of
meromorphic solutions.

Through the paper we deal with initial data subject to


\begin{condition}
\label{hyp1.1} $q$ is real and $L_{\limfunc{loc}}^{1}$ such that

\begin{enumerate}
\item (semiboundedness from below) 
\begin{equation}  \label{Cond1}
\inf\func{Spec}\left(-\partial_x^2+q(x)\right)=-h_0^2
\end{equation}
with some $h_0\ge0$.

\item (subexponential decay at $+\infty $) For $x$ large enough 
\begin{equation}
\int_{x}^{\infty }\left\vert q\right\vert \leq C_{q}e^{-cx^{\varepsilon }}
\label{Cond2}
\end{equation}%
with some positive $C_{q},c,\varepsilon $.
\end{enumerate}
\end{condition}

We assume that the constants $c,\varepsilon $ in \eqref{Cond2} are chosen
optimal.

Note that the set of such functions is very large. Indeed, in terms of $q$
itself, Condition \eqref{Cond1} is satisfied if 
\begin{equation}
\limfunc{Sup}\limits_{x}\int_{x-1}^{x}\max \left( -q,0\right) <\infty ,
\label{Cond_q}
\end{equation}%
i.e. $q$ is essentially bounded from below \cite{Glazman66}. The condition (%
\ref{Cond_q}) cannot be improved since (\ref{Cond_q}) becomes also necessary
for (\ref{Cond1}) if $q$'s are negative. Therefore, any $q$ subject to
Hypothesis \ref{hyp1.1} is essentially bounded from below, has
subexponential decay at $+\infty $ and arbitrary otherwise. Such functions
can grow (arbitrarily fast) at $-\infty $ or look like a stock market
(Gaussian white noise on a left half line) but still satisfy our hypothesis
as long as they exhibit rapid decay (\ref{Cond2}) at $+\infty $. In spectral
terms (\ref{Cond2}) implies that $\left( 0,\infty \right) $ belongs to the
absolutely continuous spectrum of $-\partial _{x}^{2}+q(x)$.

We now state our main results.


\begin{theorem}
\label{thm1.2} Under Hypothesis \ref{hyp1.1} with $\varepsilon \geq 1/2$ on
the initial data $q$ in \eqref{eq1.1} , the problem \eqref{eq1.1} has an
analytic in $z$ solution $u(z,t)$ given by 
\begin{equation}
u(z,t)=-2\partial _{z}^{2}\log \det \left( 1+\mathbb{M}(z,t)\right) ,
\label{det_form}
\end{equation}%
where $\mathbb{M}(z,t)$ is a trace class operator-valued function
constructed in Proposition \ref{pr4.1} below for any $t>0$. Moreover, for
any $t>0$

\begin{enumerate}
\item If $\varepsilon>1/2$ then $u(z,t)$ is meromorphic on $\mathbb{C}$.

\item If $\varepsilon =1/2$ then $u(z,t)$ is meromorphic in the strip 
\begin{equation}
\left\vert \Im z\right\vert <\frac{9\sqrt{2}}{8}c\sqrt{t}  \label{strip}
\end{equation}%
where $c$ is as in \eqref{Cond2}.
\end{enumerate}
\end{theorem}


\begin{theorem}
\label{thm1.3} Under Hypothesis \ref{hyp1.1} with $0<\varepsilon <1/2$ on
the initial data $q$ in \eqref{eq1.1}, the operator-valued function $\mathbb{%
M}(x,t)$ given in Proposition \ref{pr4.1} is trace class for any real $x$
and $t>0$ and 
\begin{equation*}
\mathbb{M}(x,t)=\mathbb{M}^{\left( 1\right) }(x,t)+\mathbb{M}^{\left(
2\right) }(x,t),
\end{equation*}%
where $\mathbb{M}^{\left( 1\right) }(x,t)$ is meromorphic in $x$ and $%
\mathbb{M}^{\left( 2\right) }(x,t)$ is Gevrey $G^{\frac{1}{2\varepsilon }-1}$
regular. If in addition $1+\mathbb{M}(x,t)$ is invertible for any real $x$
and $t>0$ then the problem \eqref{eq1.1} has a solution $u(x,t)$ given by 
\begin{equation}
u(x,t)=-2\partial _{x}^{2}\log \det \left( 1+\mathbb{M}(x,t)\right) ,
\label{det_form1}
\end{equation}%
belonging to the Gevrey class $G_{\limfunc{loc}}^{\frac{1}{2\varepsilon }-1}$%
.
\end{theorem}

Theorems \ref{thm1.2} and \ref{thm1.3} significantly improve our results in 
\cite{Ry11} which in turn improve Tarama \cite{Tarama04}. Theorems \ref%
{thm1.2} and \ref{thm1.3} have also some important corollaries. We will come
back to the relevant discussions in the last section when we have the
necessary background. We only mention here that our approach is based on the
Inverse Scattering Transform (IST)\ combined with pseudo-analytic
continuation techniques developed by E.M. Dyn'kin (see e.g. \cite{Dyn76}, 
\cite{BorDyn93}) and we do not believe that any of the statements of Theorem %
\ref{thm1.2} can be obtained by purely PDE techniques.

The paper is organized as follows. In Section 2, for the reader's
convenience we list our main notation and give the relevant preliminaries.
In Section 3 we define a suitable reflection coefficient and investigate its
properties which will play a central role in our consideration. The results
of this section may have some independent interest. In Section 4 we give a
brief review of the classical IST stated in terms of Hankel operators and
further prepare to prove our main results in Section 5. Section 6, the last
one, is devoted to discussions of our results and some corollaries which
directly follow from them. It also contains some open problems.


\section{Notation and Preliminaries}

We adhere to standard terminology accepted in Analysis. Namely, $\mathbb{R}%
_{\pm }\equiv[0,\pm \infty )$, $\mathbb{C}$ is the complex plane, 
\begin{equation*}
\mathbb{C}_{\pm }=\left\{ z\in \mathbb{C}:\pm \Im z>0\right\} .
\end{equation*}%
Through the paper the subscript $\pm $ indicates objects (functions,
operators, spaces, etc.) somehow related to $\mathbb{R}_{\pm }$ or $\mathbb{C%
}_{\pm }$. The bar $\overline{z}$ denotes the complex conjugate of $z$.

When appropriate, we write%
\begin{equation*}
y\eqsim x\text{ in place of \ }y=\limfunc{const}\cdot x
\end{equation*}%
and similarly whenever convenient 
\begin{equation*}
y\lesssim _{a}x\text{ in place of }y\leq C_{a}x
\end{equation*}%
with some $C_{a}>0$ dependent on a parameter $a$ but independent of $x$. If $%
C_{a}$ is an absolute constant we then write $y\lesssim x$.\ \ This will
help us keep bulky formulas under control.

We use $\left\Vert \cdot \right\Vert _{X}$ to denote the norm in a Banach
(Hilbert) space $X$.

We will need the Gevrey classes $G^{\alpha },\alpha >0,$ on $\mathbb{R}$ of
all functions $f$:

\begin{equation*}
\left\Vert \partial _{x}^{n}f\right\Vert _{L^{\infty }}\lesssim
_{f}Q_{f}^{n}\left( n!\right) ^{1+\alpha },n=0,1,2,...
\end{equation*}%
with some $Q_{f}>0.$

By \cite{BorDyn93}, Theorem 3, the statement $f\in G^{\alpha }$ is
equivalent to the statement that $f$ admits a pseudo analytic extension to
the whole complex plane such that 
\begin{equation}
\left\vert \partial _{\overline{z}}f\right\vert \lesssim _{f}\exp \left\{
-Q\left\vert \Im z\right\vert ^{-\frac{1}{\alpha }}\right\}
\label{lambda_bar}
\end{equation}%
with some $Q>0$.

In a similar manner one introduces local Gevrey classes $G_{\limfunc{loc}%
}^{\alpha }$.

Next, $\mathfrak{S}_{2}$ denotes the Hilbert-Schmidt class%
\begin{equation*}
\mathfrak{S}_{2}=\left\{ A:\left\Vert A\right\Vert _{\mathfrak{S}%
_{2}}^{2}\equiv\func{tr}\left( A^{\ast }A\right) <\infty \right\}
\end{equation*}%
and $\mathfrak{S}_{1}$ is the trace class: 
\begin{equation*}
\mathfrak{S}_{1}=\left\{ A:\left\Vert A\right\Vert _{\mathfrak{S}_{1}}\equiv%
\func{tr}\left( A^{\ast }A\right) ^{1/2}<\infty \right\} .
\end{equation*}
Note that $A\in\mathfrak{S}_1$ if and only if $A=A_1A_2$ with some $%
A_1,A_2\in\mathfrak{S}_2$.

Some other miscellaneous notation: $\chi _{S}\left( x\right) $ is the
characteristic function of a set $S$, i.e. 
\begin{equation*}
\chi _{S}\left( x\right) \equiv\left\{ 
\begin{array}{c}
1,x\in S \\ 
0,x\notin S%
\end{array}%
\right. .
\end{equation*}%
In particular $\chi _{\pm }\equiv\chi _{_{\mathbb{R}_{\pm }}}$is the
Heaviside function of $\mathbb{R}_{\pm }$. We also write 
\begin{equation*}
\left. f\right\vert _{S}=\chi _{S}f.
\end{equation*}

The notation $H_{q}\equiv -\partial _{x}^{2}+q(x)$ for the Schr\"odinger
operator on $L^{2}\left( \mathbb{R}\right) $ will be frequently used.


\section{The reflection coefficient and its analytic structure}

In this section we define a suitable reflection coefficient and investigate
its properties which will play a central role in our consideration. The
results of this section may have some independent interest.

In the short-range scattering for the full line Schr\"{o}dinger operator,
one typically introduces the right and left reflection coefficients $%
R(\lambda ),L(\lambda )$ and the transmission coefficient $T(\lambda )$ as
functions of the momentum $\lambda $ (see e.g. \cite{Deift79}). These
quantities (also called transition coefficients) can also be properly
defined in much larger spectral situations through Wronskians and/or
Titchmarsh-Weyl $m$-functions (see e.g. \cite{GNP97,GS97}). Such extensions
need not be unique. However, in our setting of step-like potentials decaying
at $+\infty $, there is a natural candidate for the right reflection
coefficient $R(\lambda )$.


\begin{definition}[\protect\cite{Ry11}]
\label{def3.1} Let $q(x)$ be real, locally integrable such that $q\in
L^{1}\left( \mathbb{R}_{+}\right) $ and $-\partial _{x}^{2}+q(x)$ is in the
limit point case at $-\infty $. Denoting by $R_{n}(\lambda )$ the right
reflection coefficient (which is necessarily well defined) from the
potential $q_{n}=q|_{(-n,\infty )}$, we call the weak limit (if it exists) 
\begin{equation}
R(\lambda )\equiv \text{w-}\lim R_{n}(\lambda ),\;n\rightarrow \infty ,
\label{eq3.1}
\end{equation}%
the right reflection coefficient from the potential $q$.
\end{definition}

Note that one should not expect in \eqref{eq3.1}\ pointwise convergence as
an explicit counterexample $q=$ $\chi _{-}$ readily shows. Uniform
convergence in \eqref{eq3.1} is not available in general even in the
short-range setting \cite{Deift79}.

As shown in \cite{Ry11}, Lemma 5.4, the reflection coefficient introduced
this way is well defined. The following statement will play a crucial role
in our consideration.


\begin{proposition}[the analytic structure of the reflection coefficient]
\label{pr3.2} Under Hypothesis \ref{hyp1.1}, the right reflection
coefficient given by \eqref{eq3.1} exists and admits the representation 
\begin{equation}
R(\lambda )=A(\lambda )+\frac{S(\lambda )G(\lambda )}{\lambda B(\lambda )}
\label{Rep}
\end{equation}%
where functions $A,B,S,G$ have the properties

\begin{enumerate}
\item \label{it1} $A$ is an analytic on $\mathbb{C}^+\setminus[0,ih_0]$
function such that $\left\vert A \right\vert \le2$ on $\mathbb{R}$ and $%
A(\lambda)=o\left(1/\lambda\right)$, $\lambda\to\infty$ along any ray in $%
\mathbb{C}^+$

\item \label{it2} $B$ is the Blaschke product 
\begin{equation*}
B(\lambda )=\prod_{k=1}^{N}\frac{\lambda -i\varkappa _{k}}{\lambda
+i\varkappa _{k}}
\end{equation*}%
where real $\varkappa _{k}$'s are such that $\left\{ -\varkappa
_{k}^{2}\right\} _{k=1}^{N}$ is the negative discrete spectrum of $H_{q_{+}}$%
, $q_{+}\equiv q|_{\mathbb{R}_{+}}$

\item \label{it3} $\left\vert S(\lambda )\right\vert \leq 1$, $\lambda \in 
\mathbb{C}^{+}$

\item \label{it4} $G\in G^{\frac{1}{\varepsilon }-1}$

\item \label{it5} $\left\vert S(\lambda )G(\lambda )/\lambda \right\vert
\leq 1$ a.e. on $\mathbb{R}$

\item \label{it6} If $R_{n}$ is as in Definition \ref{def3.1} then 
\begin{equation*}
R_{n}(\lambda )=A_{n}(\lambda )+\frac{S(\lambda )G(\lambda )}{\lambda
B(\lambda )}
\end{equation*}%
and 
\begin{equation*}
A_{n}\rightarrow A,\;n\rightarrow \infty
\end{equation*}%
uniformly on any compact in $\mathbb{C}^{+}\setminus \lbrack 0,ih_{0}]$.
\end{enumerate}
\end{proposition}

\begin{proof}
Most of statements in Proposition \ref{pr3.2} (save \eqref{it4}) are proven
in \cite{Ry11} and we restrict ourselves to some comments only. Note first
that Condition 1 of Hypothesis \ref{hyp1.1} implies that $-\partial
_{x}^{2}+q(x)$ is in the limit point case at $-\infty $ (see, e.g. \cite%
{ClarkGeszt03} for complete results on this matter). Splitting 
\begin{equation}
q=q_{-}+q_{+},\ \ q_{\pm }=q|_{\mathbb{R}_{\pm }}  \label{split_q}
\end{equation}%
induces the representation 
\begin{equation*}
R=\frac{T_{+}^{2}R_{-}}{1-R_{-}L_{+}}+R_{+}
\end{equation*}%
where $\pm $ label scattering quantities associated with $q_{\pm }$. The
functions $T_{+},L_{+},R_{-}$ can be analytically continued into $\mathbb{C}%
^{+}$ and 
\begin{equation*}
A\equiv \frac{T_{+}^{2}R_{-}}{1-R_{-}L_{+}}
\end{equation*}%
has properties \eqref{it1}, \eqref{it6}. For $R_{+}$, which is independent
of $n$, we use the representation \cite{Deift79}, Theorem 2, 
\begin{equation*}
R_{+}(\lambda )=\frac{T_{+}(\lambda )}{\lambda }G(\lambda )
\end{equation*}%
where 
\begin{equation}
G(\lambda )=\frac{1}{2i}\int_{-\infty }^{\infty }e^{-2i\lambda x}g(x)dx
\label{eq6'.1}
\end{equation}%
with some $g$ obeying 
\begin{equation}
\left\vert g(x)\right\vert \leq \left\vert q(x)\right\vert +\limfunc{const}%
\int_{x}^{\infty }\left\vert q\right\vert .  \label{Est_on_g}
\end{equation}%
Since $R_{+}(\lambda )$ is a reflection coefficient we have \eqref{it5}.
Since $T_{+}$ is a transmission coefficient, 
\begin{equation*}
T_{+}(\lambda )=\prod_{k=1}^{N}\frac{\lambda +i\varkappa _{k}}{\lambda
-i\varkappa _{k}}\cdot S(\lambda )=B\left( \lambda \right) ^{-1}S(\lambda )
\end{equation*}%
where $S$ is an outer function of $\mathbb{C}^{+}$: $\left\vert S(\lambda
)\right\vert \leq 1$, $\lambda \in \mathbb{C}^{+}$. This proves \eqref{it2}
and \eqref{it3}.

The proposition is proven if we show \eqref{it4}. Due to \eqref{lambda_bar}
we should demonstrate that $G$ admits a pseudo analytic extension the whole
complex plane such that 
\begin{equation}
\left\vert \partial _{\overline{\lambda }}G\right\vert \lesssim \exp \left\{
-Q\left\vert \Im \lambda \right\vert ^{-\frac{\varepsilon }{1-\varepsilon }%
}\right\}  \label{eq6''.1}
\end{equation}%
with some $Q>0$. There are a few explicit ways to construct pseudo analytic
continuations (see e.g. \cite{Dyn76}, \cite{BorDyn93}, \cite{Tarama04})
producing different extensions. We modify the one used in \cite{Tarama04} to
obtain a better $Q$ in \ref{eq6''.1}. Note that 
\begin{equation}
G\left( \lambda \right) \eqsim \widehat{g}\left( 2\lambda \right)  \label{Gg}
\end{equation}%
where $\widehat{g}$ is the Fourier transform of $g$ which due to %
\eqref{Est_on_g} satisfies Condition 2 of Hypothesis \ref{hyp1.1} with some $%
\widetilde{c}<c$. I.e. 
\begin{equation}
\int_{x}^{\infty }\left\vert g\right\vert \lesssim _{g}e^{-\widetilde{c}%
x^{\varepsilon }}.  \label{est_g}
\end{equation}%
For any $\lambda \in \mathbb{C}^{+}$ define 
\begin{equation}
\widetilde{G}\left( \lambda \right) =\sum_{n\geq 1}\theta \left(
r^{\varepsilon +2}x_{n}^{1-\varepsilon }\frac{\Im \lambda }{\widetilde{c}}%
\right) G_{n}(\lambda ),  \label{eq6''.1too}
\end{equation}%
where $\theta $ is a smooth on $\mathbb{R}_{+}$ function such that: 
\begin{align*}
\theta (x)& =1,\;x\in \lbrack 0,1], \\
\theta (x)& =0,\;x\geq r,
\end{align*}%
$r>1$, $x_{n}=r^{n}$ and 
\begin{equation*}
G_{n}(\lambda )=\int_{x_{n-1}}^{x_{n}}e^{-i\lambda x}g(x)dx.
\end{equation*}%
The formula \eqref{eq6''.1too} clearly defines an extension of $\widehat{g}%
\left( \lambda \right) $ to complex $\lambda $. We next show that $%
\widetilde{G}$ is uniformly bounded on $\mathbb{C}^{+}$. Bound $G_{n}$
first. By \eqref{est_g}%
\begin{equation*}
\left\vert G_{n}(\lambda )\right\vert \lesssim e^{\left\vert \Im \lambda
\right\vert \cdot x_{n}}\int_{x_{n-1}}^{x_{n}}\left\vert g\right\vert
\lesssim _{g}\exp \left\{ \left\vert \Im \lambda \right\vert \cdot x_{n}-%
\widetilde{c}x_{n-1}^{\varepsilon }\right\}
\end{equation*}%
and one has 
\begin{equation}
\left\vert \widetilde{G}(\lambda )\right\vert \lesssim _{g}\sum_{n\geq
1}\sum_{n\geq 1}\theta \left( r^{\varepsilon +2}x_{n}^{1-\varepsilon }\frac{%
\Im \lambda }{\widetilde{c}}\right) \exp \left\{ \left\vert \Im \lambda
\right\vert \cdot x_{n}-\widetilde{c}x_{n-1}^{\varepsilon }\right\} .
\label{eq6'''.1}
\end{equation}
In \eqref{eq6'''.1} many terms are in fact zero and nontrivial ones are
subject to 
\begin{equation*}
r^{\varepsilon +2}x_{n}^{1-\varepsilon }\frac{\left\vert \Im \lambda
\right\vert }{\widetilde{c}}\leq r.
\end{equation*}%
I.e. only nonzero terms in \eqref{eq6'''.1} are the ones obeying 
\begin{equation}
x_{n}^{1-\varepsilon }\leq \frac{\widetilde{c}}{r^{\varepsilon +1}}\cdot 
\frac{1}{\left\vert \Im \lambda \right\vert }.  \label{eq6'''.2}
\end{equation}

Under the condition \eqref{eq6'''.2}, for the argument of the exponential in %
\eqref{eq6'''.1}, we have ($1/r<\delta <1$) 
\begin{align}
\left\vert \Im \lambda \right\vert \cdot x_{n}-\widetilde{c}r^{-\varepsilon
}x_{n}^{\varepsilon }& =\left( \left\vert \Im \lambda \right\vert \cdot
x_{n}-\delta \widetilde{c}r^{-\varepsilon }x_{n}^{\varepsilon }\right)
-(1-\delta )r^{-\varepsilon }x_{n}^{\varepsilon }  \label{eq6'''.3} \\
& =\left\vert \Im \lambda \right\vert x_{n}^{\varepsilon }\left(
x_{n}^{1-\varepsilon }-\delta \frac{\widetilde{c}}{\left\vert \Im \lambda
\right\vert }\right) -(1-\delta )r^{-\varepsilon }x_{n}^{\varepsilon }. 
\notag
\end{align}%
By \eqref{eq6'''.2} the right hand side of \eqref{eq6'''.3} doesn't exceed 
\begin{eqnarray*}
&&\left\vert \Im \lambda \right\vert x_{n}^{\varepsilon }\left( \frac{%
\widetilde{c}}{\left\vert \Im \lambda \right\vert }\frac{1}{r^{\varepsilon
+1}}-\frac{\widetilde{c}}{\left\vert \Im \lambda \right\vert }\frac{\delta }{%
r^{\varepsilon }}\right) -(1-\delta )r^{-\varepsilon }x_{n}^{\varepsilon } \\
&=&-\widetilde{c}\left( \delta -\frac{1}{r}\right) x_{n-1}^{\varepsilon
}-(1-\delta )x_{n-1}^{\varepsilon } \\
&<&-\limfunc{const}x_{n-1}^{\varepsilon }.
\end{eqnarray*}%
It follows now from this estimate and \eqref{eq6'''.1} that 
\begin{equation}
\left\vert \widetilde{G}(\lambda )\right\vert \lesssim _{g}\sum_{n\geq
0}\exp \{-\limfunc{const}x_{n}^{\varepsilon }\}<\infty .  \label{eq6iv.0}
\end{equation}%
Similarly one proves that all derivatives of $G$ are also bounded on $%
\mathbb{C}^{+}$.

It remains now to show \eqref{eq6''.1too}. One has 
\begin{align}
\left\vert \partial _{\overline{\lambda }}\widetilde{G}\right\vert & \leq
\sum_{n\geq 1}\theta ^{\prime }\left( r^{\varepsilon +2}x_{n}^{1-\varepsilon
}\frac{\left\vert \Im \lambda \right\vert }{\widetilde{c}}\right) \frac{%
r^{\varepsilon +1}x_{n}^{1-\varepsilon }}{2\widetilde{c}}\left\vert
G_{n}\right\vert  \label{eq6iv.1} \\
& \lesssim _{g}\sum_{n\geq 1}x_{n}^{1-\varepsilon }\exp \{\left\vert \Im
\lambda \right\vert \cdot x_{n}-\widetilde{c}r^{-\varepsilon
}x_{n}^{\varepsilon }\}.  \notag
\end{align}
Only terms subject to 
\begin{equation}
\frac{\widetilde{c}r^{-\varepsilon -2}}{\left\vert \Im \lambda \right\vert }%
\leq x_{n}^{1-\varepsilon }\leq \frac{\widetilde{c}r^{-\varepsilon -1}}{%
\left\vert \Im \lambda \right\vert }  \label{eq6iv.2}
\end{equation}%
make a non trivial contribution to the series in \eqref{eq6iv.1}. The
inequality \eqref{eq6iv.2} implies 
\begin{align}
x_{n}& \geq \left( \frac{\widetilde{c}r^{-\varepsilon -2}}{\left\vert \Im
\lambda \right\vert }\right) ^{\frac{1}{1-\varepsilon }},  \notag \\
\intertext{or}
x_{n}^{\varepsilon }& \geq \left( \frac{\widetilde{c}r^{-\varepsilon -2}}{%
\left\vert \Im \lambda \right\vert }\right) ^{\frac{\varepsilon }{%
1-\varepsilon }}.  \label{eq6iv.3}
\end{align}

Splitting the argument of the exponential in \eqref{eq6iv.1} same way as %
\eqref{eq6'''.3} and using \eqref{eq6iv.3}, we have 
\begin{align*}
\left\vert \Im \lambda \right\vert \cdot x_{n}^{\varepsilon }-\frac{%
\widetilde{c}}{r^{\alpha }}x_{n}^{\alpha }& \leq \left\vert \Im \lambda
\right\vert \cdot x_{n}^{\varepsilon }\left( \frac{\widetilde{c}%
r^{-\varepsilon -1}}{\left\vert \Im \lambda \right\vert }-\frac{%
r^{-\varepsilon }\delta \widetilde{c}}{\left\vert \Im \lambda \right\vert }%
\right) -(1-\delta )\frac{\widetilde{c}x_{n}^{\varepsilon }}{r^{\varepsilon }%
} \\
& =-x_{n}^{\varepsilon }\widetilde{c}r^{-\varepsilon -1}(r\delta
-1)-(1-\delta )\frac{\widetilde{c}x_{n}^{\varepsilon }}{r^{\varepsilon }} \\
& \leq -\left( \frac{\widetilde{c}r^{-\varepsilon -2}}{\left\vert \Im
\lambda \right\vert }\right) ^{\frac{\varepsilon }{1+\varepsilon }}%
\widetilde{c}r^{-\varepsilon -1}(r\delta -1)-(1-\delta )\frac{\widetilde{c}%
x_{n}^{\varepsilon }}{r^{\varepsilon }} \\
& -\frac{\widetilde{c}^{\frac{1}{1-\varepsilon }}}{r^{\frac{2\varepsilon +1}{%
1-\varepsilon }}}\frac{r\delta -1}{\left\vert \Im \lambda \right\vert ^{%
\frac{\varepsilon }{1-\varepsilon }}}-(1-\delta )\frac{\widetilde{c}%
x_{n}^{\varepsilon }}{r^{\varepsilon }}.
\end{align*}

Inserting this into \eqref{eq6iv.1} we obtain 
\begin{align}
\left\vert \partial _{\overline{\lambda }}\widetilde{G}\right\vert &
\lesssim _{g}\left( \sum_{n\geq 0}x_{n}^{1-\varepsilon }\exp \{-\limfunc{%
const}x_{n}^{\varepsilon }\}\right) \cdot \exp \left\{ -\widetilde{Q}%
\left\vert \Im \lambda \right\vert ^{-\frac{\varepsilon }{1-\varepsilon }%
}\right\}  \label{eq6v.0} \\
\widetilde{Q}& \equiv(r\delta -1)\frac{\widetilde{c}^{\frac{1}{1-\varepsilon 
}}}{r^{\frac{2\varepsilon +1}{1-\varepsilon }}}<(r\delta -1)\frac{c^{\frac{1%
}{1-\varepsilon }}}{r^{\frac{2\varepsilon +1}{1-\varepsilon }}}.
\label{eq6v.1}
\end{align}%
The series in \eqref{eq6v.0} is convergent and $\widetilde{G}\left( \lambda
\right) $ is an pseudo analytic extension of $\widehat{g}\left( \lambda
\right) $ from the real line to the upper half plane. Due to \eqref{Gg} we
have found a pseudo analytic extension of $G$ subject to \eqref{eq6''.1}
with $Q=2\widetilde{Q}$. This completes our proof.
\end{proof}


\begin{remark}
The representation \eqref{Rep} is not unique. It depends on the reference
point in the splitting of \eqref{split_q}. This flexibility will be used
later.
\end{remark}


\begin{remark}
We have also had some flexibility in choosing $r$ and $\delta $ in %
\eqref{eq6v.1} subject to $r>1$, $1/r<\delta <1$. The range for $Q=2%
\widetilde{Q}$ given by \eqref{eq6v.1} is 
\begin{equation*}
0<Q<\frac{2(\varepsilon -1)(3\varepsilon )^{\frac{2\varepsilon +1}{%
1-\varepsilon }}}{(2\varepsilon +1)^{\frac{1+\varepsilon }{1-\varepsilon }}}%
\cdot c^{\frac{1}{1-\varepsilon }}
\end{equation*}%
which is inessential to what follows but the borderline case $\varepsilon
=1/2$. For this case, 
\begin{equation}
0<Q<\frac{3^{3}}{2^{7}}c^{2}.  \label{eq7.1}
\end{equation}
\end{remark}


\section{The Marchenko integral operator and the inverse scattering transform%
}

The integral operator we are concerned with in this section appears to have
been introduced by Marchenko and received a comprehensive treatment in his
classical book \cite{Mar86}. To acknowledge Marchenko's profound
contribution to the subject, we denote this operator by $\mathbb{M}$
(Marchenko used $\mathbb{F}$) but otherwise try to retain as much of his
original notation as possible.

We call the integral operator $\mathbb{M}:L^{2}\left( \mathbb{R}_{+}\right)
\rightarrow L^{2}\left( \mathbb{R}_{+}\right)$ Marchenko if 
\begin{align}
\left( \mathbb{M}f\right) (x)& =\int_{0}^{\infty }M(x+y)f(y)dy,\quad f\in
L^{2}\left( \mathbb{R}_{+}\right) ,  \notag \\
M(\cdot )& =\int_{0}^{\infty }e^{-(\cdot )\lambda }d\rho (\lambda
)+\int_{-\infty }^{\infty }e^{2i(\cdot )\lambda }R(\lambda )\frac{d\lambda }{%
2\pi },  \label{eq4.1}
\end{align}%
where $\rho $ is a finite nonnegative measure and $R$ is such that for a.e. $%
\lambda \in \mathbb{R}$ 
\begin{equation*}
R(-\lambda )=\overline{R(\lambda )}\quad ,\quad \left\vert R(\lambda
)\right\vert \leq 1.
\end{equation*}

The operator $\mathbb{M}$ introduced this way is clearly a Hankel operator
(the kernel depends on the sum of the arguments). We say that $\mathbb{M}$
is associated with a potential $q$ if $R$ is a reflection coefficient from $%
q $ and $\rho $ characterizes the negative spectrum of $H_{q}$. For
short-range $q$'s, the measure $\rho $ is purely discrete and ($\delta $ is
the Dirac $\delta $-function) 
\begin{equation*}
d\rho (\lambda )=\sum_{n=1}^{N}c_{n}^{2}\delta (\lambda -\varkappa
_{n})d\lambda
\end{equation*}%
where $\varkappa _{n}$'s are such that $\left\{ -\varkappa _{n}^{2}\right\}
_{n=1}^{N}$ are the negative bound states of $H_{q}$ and $%
\left\{c_{n}\right\}_{n=1}^N$ are the associated norming constants. Of
course, the kernel $M$ (which we also call Marchenko) is nothing but the sum
of the Laplace transform of the (finite, positive) measure $\rho $ and the
Fourier transform of the (symmetric, bounded) function $R$. This is the main
feature of the Marchenko operator resulting in the decomposition 
\begin{equation}
\mathbb{M}=\mathbb{M}_{1}+\mathbb{M}_{2},  \label{eq4.2}
\end{equation}%
where $\mathbb{M}_{1}\geq 0$ and ($\chi =\chi _{+}$, $\mathcal{F}$ is the
Fourier transform) 
\begin{equation*}
\mathbb{M}_{2}=\chi \mathcal{F}^{-1}R\mathcal{F}^{-1}
\end{equation*}%
and is selfadjoint.

Note that the Marchenko operator is not typically studied in the context of
Hankel operators. We have found in \cite{Ry11} that the language of Hankel
operators is very convenient in the IST formalism. In this language, $\rho
,R $ are called the measure and symbol of $\mathbb{M}_{1},\mathbb{M}_{2}$
respectively (see e.g. \cite{Nik2001}). In the language of inverse
scattering, $(\rho ,R)$ are commonly referred to as the (right) scattering
data\footnote{%
For short-range $q$'s, the so-called left scattering data are also
considered, which need not be well-defined in our setting.}.

In the context of the Cauchy problem for the KdV equation, we have a two
parametric family of Marchenko operators $\mathbb{M}(z,t)$, where $(z,t)$
are the (real) variables in \eqref{eq1.1}. Namely, the measure $\rho _{z,t}$
and the symbol $R_{z,t}$ of $\mathbb{M}(z,t)$ are given by 
\begin{align*}
d\rho _{z,t}(\lambda )& =\zeta _{z,t}(i\lambda )d\rho (\lambda ), \\
R_{z,t}(\lambda )& =\zeta _{z,t}(\lambda )R(\lambda ), \\
\zeta _{z,t}(\lambda )& :\equiv e^{8i\lambda ^{3}t+2i\lambda z}
\end{align*}%
and $(\rho ,R)$ are the scattering data for the profile $q$. Clearly for $%
\lambda \in \mathbb{R}$ 
\begin{align*}
\left\vert \zeta _{z,t}(\lambda )\right\vert & =1, \\
\left\vert \zeta _{z,t}(i\lambda )\right\vert & =e^{8\lambda ^{3}t-2\lambda
z}.
\end{align*}

The whole point of the IST is that $(\rho _{0,t},R_{0,t})$ are the
scattering data for $H_{u(z,t)}$ where $u(z,t)$ solves \eqref{eq1.1}. The
actual mechanism to recover $u(z,t)$ amounts to solving the Marchenko
integral equation\footnote{%
also known as Gelfand-Levitan-Marchenko.} or equivalently through the
Riemann-Hilbert problem. For our purposes the explicit formula %
\eqref{det_form} is convenient. Note that \eqref{det_form} is nothing but
Cramer's rule for linear integral equations. For $R\equiv 0$ (reflectionless
initial profile), assuming that $q$ is short-range, $\mathbb{M}(z,t)$ turns
into a finite rank operator of a very explicit structure. The formula %
\eqref{det_form} in this particular case was discovered in the 1960s. In the
general short-range case ($R\neq 0$), we are not sure whom it should be
attributed to but it was systemically studied by P\"{o}ppe in the 1980s
(see, e.g. \cite{P84} and also \cite{P89} for the Kadomtsev-Petviashvili and
\ \cite{P83} for the sine-Gordon equations. In the context of nondecaying
initial data, it appears first in \cite{Ry11}. Note that the sense in what
the determinant in \eqref{det_form} is understood is an issue which doesn't
seem to be fully addressed in the literature. It is typically defined by the
Fredholm formula through an absolutely convergent (Fredholm) series. We
actually show that $\mathbb{M}(z,t)$ is trace class for any $z$ (even
complex) and $t>0$ under Hypothesis \ref{hyp1.1}. This means that $\det
\left( 1+\mathbb{M}(z,t)\right) $ is an invariant, i.e. it produces the same
value in any basis in $L^{2}\left( \mathbb{R}_{+}\right) $.

In the setting of step-like potentials, the Marchenko operator has been
intensively studied in the Kharkov mathematical school by Hruslov, Kotlyarov
and their students\footnote{%
Remark that this school has been greatly infuenced by Marchenko himself and
he remains to be its part.} (see, e.g. \cite{Hruslov76}, \cite{KhrKot94}).

We also refer to Cohen \ \cite{Cohen1984}, Kappeler \cite{Kappeler86},
Venakides \cite{Ven86} (and the literature cited therein), and recent
Egorova-Teschl \cite{ET11}. In all the above papers save \cite{ET11}, $q$'s
are assumed to have a specific type of behavior at $-\infty $ (approaching
either a constant or a periodic function) and fall off at $+\infty$. In \cite%
{ET11}, the interesting case of two finite gap potentials fused together is
considered.

We summarize important properties of the Marchenko operator in the following
(see \cite{Ry11} for details).


\begin{proposition}[The structure of the Marchenko operator]
\label{pr4.1} Assuming Hypothesis \ref{hyp1.1}, let $\mathbb{M}(z,t)$ be the
Marchenko operator associated with $q$ and let $A$ be as in Proposition \ref%
{pr3.2}. Then for any $z\in \mathbb{R}$, $t>0$, 
\begin{equation}
\mathbb{M}(z,t)=\mathbb{M}_{+}(z,t)+\mathbb{A}(z,t),  \label{eq4}
\end{equation}%
where $\mathbb{M}_{+}(z,t)$ is the Marchenko operator associated with $%
q_{+}=q|_{\mathbb{R}_{+}}$ and $\mathbb{A}(z,t)$ is a Hankel integral
operator with the kernel 
\begin{equation*}
\frac{1}{2\pi }\int_{\mathbb{R}+ih}e^{2i\lambda (\cdot )}\zeta
_{z,t}(\lambda )A(\lambda )d\lambda ,\quad h>h_{0}.
\end{equation*}%
Furthermore, $\mathbb{A}(z,t)$ is an entire operator-valued function of
trace class for any complex $z$ and $t>0$, continuous with respect to $q$ in
the following sense: If $q_{1},q_{2}$ are two functions subject to
Hypothesis \ref{hyp1.1} then 
\begin{equation*}
\left\Vert \mathbb{A}_{1}(z,t)-\mathbb{A}_{2}(z,t)\right\Vert _{\mathfrak{S}%
_{1}}\leq \frac{1}{4\pi h}\left\Vert \zeta _{z,t}(A_{1}-A_{2})\right\Vert
_{L^{1}(\mathbb{R}+ih)}
\end{equation*}%
for any $z\in \mathbb{C}$, $t>0$.
\end{proposition}

Note that $\mathbb{M}(z,t)$ depends on $(z,t)$ through $\zeta _{z,t}$.


\section{Proof of the main results}

With all the preparations done in the previous sections, the actual proofs
will be quite short. It is convenient to conduct both proofs at a time. Note
first that, by a trivial shifting, we may assume without loss of generality
that $H_{q_{+}}$ has at most one bound state $-\varkappa ^{2}$. Consider the
problem \eqref{eq1.1} with 
\begin{equation*}
q_{n}(x)=%
\begin{cases}
q(x)\quad & ,\quad x\geq -n \\ 
0 & ,\quad x<-n%
\end{cases}%
.
\end{equation*}

It is well-known that for such initial profiles\footnote{%
So far we only know that the determinant exists in the Fredholm sense.} 
\begin{equation}
u_{n}(z,t)=-2\partial _{z}^{2}\log \det \left( 1+\mathbb{M}_{n}(z,t)\right) .
\label{eq5.1}
\end{equation}%
By Proposition \ref{pr4.1} 
\begin{equation*}
\mathbb{M}(z,t)=\mathbb{M}_{+}(z,t)+\mathbb{A}(z,t)+\delta \mathbb{A}(z,t)
\end{equation*}%
where $\delta \mathbb{A}\equiv\mathbb{A}_{n}-\mathbb{A}$ is meromorphic in $%
z $ for any $t>0$ and small in the $\mathfrak{S}_{1}$-norm for $n$ large
enough. I.e. 
\begin{equation}
\left\Vert \mathbb{M}_{n}(z,t)-\mathbb{M}(z,t)\right\Vert _{\mathfrak{S}%
_{1}}\rightarrow 0,\quad n\rightarrow \infty .  \label{eq5.2.1}
\end{equation}%
Therefore, $\mathbb{M}(z,t)\in \mathfrak{S}_{1}$ is proven if we show that $%
\mathbb{M}_{+}(z,t)\in \mathfrak{S}_{1}$.

Split 
\begin{equation*}
\mathbb{M}_{+}(z,t)=\mathbb{M}_{1}^{+}(z,t)+\mathbb{M}_{2}^{+}(z,t)
\end{equation*}%
where $\mathbb{M}_{1}^{+}(z,t),\mathbb{M}_{2}^{+}(z,t)$ are the Hankel
operators with the kernels 
\begin{equation*}
c_{0}^{2}\zeta _{z,t}(i\varkappa )e^{-\varkappa (x+y)},
\end{equation*}%
and 
\begin{equation*}
\frac{1}{2\pi }\int_{-\infty }^{\infty }e^{i\lambda (x+y)}\zeta
_{z,t}(\lambda )R_{+}(\lambda )d\lambda
\end{equation*}%
respectively. Here $c_{0}$ stands for the norming constant associated with
the bound state $-\varkappa ^{2}$.

The operator $\mathbb{M}_{1}(z,t)$ is rank 1 and clearly entire in $z$. Thus
we only need to properly control $\partial _{z}^{n}\mathbb{M}_{2}^{+}(z,t)$
in the $\mathfrak{S}_{1}$-norm. Evaluate (so far formally) the kernel of $%
\partial _{z}^{n}\mathbb{M}_{2}^{+}(z,t),n=0,1,2,...,$ by the Green formula
applied to the strip $\mathbb{R}\times (0,\varkappa /2)$ and by Proposition %
\ref{pr3.2} ($\lambda =\alpha +i\beta ,\partial _{\overline{\lambda }}=\frac{%
1}{2}(\partial _{\alpha }+i\partial _{\beta })$) 
\begin{align}
\frac{1}{2\pi }\int_{-\infty }^{\infty }e^{i\lambda (x+y)}(2i\lambda
)^{n}\zeta _{z,t}(\lambda )R_{+}(\lambda )d\lambda & =\int_{\mathbb{R}%
+i\varkappa /2}e^{i\lambda (x+y)}(2i\lambda )^{n}\left( B^{-1}SG\right)
(\lambda )\frac{d\lambda }{2\pi }  \notag \\
& \quad +2i\int_{0}^{\varkappa /2}d\beta \int d\alpha \;e^{i\lambda
(x+y)}F(\alpha ,\beta )  \label{eq5.3.1} \\
& \equiv H_{1}(x+y)+H_{2}(x+y),  \notag
\end{align}%
where 
\begin{equation*}
F(\alpha ,\beta )\equiv \frac{1}{2\pi }\zeta _{z,t}(\lambda )(2i\lambda
)^{n-1}\frac{S(\lambda )}{B(\lambda )}\partial _{\overline{\lambda }%
}G(\alpha ,\beta ).
\end{equation*}

Due to the rapid decay of $e^{8i\lambda^3t}$ as $\lambda\to\infty$ along $%
\mathbb{R}+ih$, the function $F(\alpha,\beta)$ is subject to the conditions
of Proposition \ref{pr4.1} and hence the integral operator with kernel $H_1$
is trace class. Our analysis of the integral operator with kernel $H_2$ is
based upon the following lemma.


\begin{lemma}
\label{lem5.1} Let $F(\alpha ,\beta )$ be such that for some $h>0$ 
\begin{equation}
\int_{0}^{h}\left( \int_{-\infty }^{\infty }\left\vert F\left( \alpha ,\beta
\right) \right\vert d\alpha \right) \frac{d\beta }{\beta }<\infty .
\label{star}
\end{equation}%
Then the integral Hankel operator $\mathbb{H}$ with the kernel ($\lambda
=\alpha +i\beta $) 
\begin{equation*}
H(x)=\int_{0}^{h}d\beta \int_{-\infty }^{\infty }\frac{d\alpha }{2\pi }%
e^{i\lambda x}F(\alpha ,\beta )
\end{equation*}%
is trace class and 
\begin{equation*}
\left\Vert \mathbb{H}\right\Vert _{\mathfrak{S}_{1}}\leq \frac{1}{2}%
\int_{0}^{h}\frac{d\beta }{\beta }\int_{-\infty }^{\infty }d\alpha
\left\vert F(\alpha ,\beta )\right\vert .
\end{equation*}
\end{lemma}

\begin{proof}
We have 
\begin{align}
H(x+y)& \eqsim \int_{0}^{h}d\beta e^{-\beta x}\int_{-\infty }^{\infty
}d\alpha e^{i\alpha x}F(\alpha ,\beta )  \notag \\
& =\int_{0}^{h}e^{-\beta (x+y)}\widehat{F_{\beta }}(x+y)dx,
\label{eq5.2.1too} \\
\widehat{F_{\beta }}(x+y)& \equiv\int_{-\infty }^{\infty }e^{i\alpha
(x+y)}F(\alpha +\beta )d\alpha  \notag \\
& \eqsim \widehat{F_{\beta }^{1/2}}\ast \widehat{F_{\beta }^{1/2}}(x+y) 
\notag \\
& =\int_{-\infty }^{\infty }F_{\beta }^{1/2}(x+s)F_{\beta }^{1/2}(y-s)ds.
\label{eq5.2.2}
\end{align}
Here we have used the convolution theorem. Inserting \eqref{eq5.2.2} into %
\eqref{eq5.2.1too} implies that 
\begin{equation}
\mathbb{H}=\int_{0}^{h}\mathbb{H}_{\beta ,1}\mathbb{H}_{\beta ,2}d\beta,
\label{eq5.2.3}
\end{equation}%
where $\mathbb{H}_{\beta ,1}$ and $\mathbb{H}_{\beta ,2}$ and integral (but
not Hankel) operators on $L^{2}\left( \mathbb{R}\right) $ with the kernels 
\begin{align*}
H_{\beta ,1}(x,s)& =\chi (x)e^{-\beta x}\widehat{F_{\beta }^{1/2}}(x+s), \\
H_{\beta ,2}(s,y)& =\chi (y)e^{-\beta y}\widehat{F_{\beta }^{1/2}}(y-s)
\end{align*}%
respectively.

It follows from \eqref{eq5.2.3} that 
\begin{equation}  \label{eq5.3.1too}
\left\Vert \mathbb{H} \right\Vert_{\mathfrak{S}_1} \le \int_0^h \left\Vert 
\mathbb{H}_{\beta,1} \right\Vert_{\mathfrak{S}_2} \cdot\left\Vert \mathbb{H}%
_{\beta,2} \right\Vert_{\mathfrak{S}_2}d\beta.
\end{equation}

Evaluate now the Hilbert-Schmidt norms of $\mathbb{H}_{\beta ,1}$ and $%
\mathbb{H}_{\beta ,2}$. By the Plancherel equation we have 
\begin{align*}
\left\Vert \mathbb{H}_{\beta ,1}\right\Vert _{\mathfrak{S_{2}}}^{2}&
=\int_{0}^{\infty }\int_{-\infty }^{\infty }\left\vert H_{\beta
,1}(x,s)\right\vert ^{2}ds\;dx \\
& =\int_{0}^{\infty }dxe^{-2\beta x}\int_{-\infty }^{\infty }ds\left\vert 
\widehat{F_{\beta }^{1/2}}(x+s)\right\vert ^{2} \\
& =\frac{\left\Vert F_{\beta }^{1/2}\right\Vert _{L^{2}\left( \mathbb{R}%
\right) }^{2}}{2\beta }=\frac{1}{2\beta }\left\Vert F_{\beta }\right\Vert
_{L^{1}\left( \mathbb{R}\right) }.
\end{align*}%
That is 
\begin{equation*}
\left\Vert H_{\beta ,1}\right\Vert _{\mathfrak{S}_{2}}\leq \frac{1}{\sqrt{%
2\beta }}\left\Vert F_{\beta }\right\Vert _{L^{1}\left( \mathbb{R}\right)
}^{1/2}.
\end{equation*}%
Similarly, 
\begin{equation*}
\left\Vert H_{\beta ,2}\right\Vert _{\mathfrak{S}_{2}}\leq \frac{1}{\sqrt{%
2\beta }}\left\Vert F_{\beta }\right\Vert _{L^{1}\left( \mathbb{R}\right)
}^{1/2}
\end{equation*}%
and \eqref{eq5.3.1too} yields 
\begin{equation*}
\left\Vert \mathbb{H}\right\Vert _{\mathfrak{S}_{1}}\leq \frac{1}{2}%
\int_{0}^{h}\left\Vert F_{\beta }\right\Vert _{L^{1}}\frac{d\beta }{\beta }.
\end{equation*}%
The lemma is proven.
\end{proof}

Let us find suitable bounds on $\mathbb{R}\times \lbrack 0,\varkappa /2]$
for the functions involved in $F$: 
\begin{align*}
\left\vert \zeta _{z,t}(\lambda )\right\vert & =\left\vert e^{8i\lambda
^{3}t+2i\lambda z}\right\vert =e^{8\beta ^{3}t-2\beta \Re z}\cdot
e^{-24\beta t\alpha ^{2}-2\alpha \Im z} \\
& \leq e^{\varkappa (\varkappa ^{2}t+\left\vert z\right\vert )}\cdot \exp
\left\{ -\left( \sqrt{24\beta t}\alpha +\frac{\Im z}{\sqrt{24\beta t}}%
\right) ^{2}+\frac{\Im ^{2}z}{24\beta t}\right\} , \\
\left\vert \lambda ^{n-1}B^{-1}(\lambda )S(\lambda )\right\vert & \lesssim
\left( \left\vert \alpha \right\vert +\beta \right) ^{n-1}, \\
\left\vert \partial _{\overline{\lambda }}G\right\vert & \lesssim
_{q_{+}}\exp \left\{ -Q\beta ^{-\frac{\varepsilon }{1-\varepsilon }}\right\}
.
\end{align*}
Thus 
\begin{equation}
\left\vert F(\alpha ,\beta )\right\vert \lesssim _{q_{+}}e^{\varkappa \left(
\varkappa ^{2}t+\left\vert z\right\vert \right) }\left( \left\vert \alpha
\right\vert +\beta \right) ^{n-1}e^{-\left( \sqrt{24\beta t}\alpha +\frac{%
\Im z}{\sqrt{24\beta t}}\right) ^{2}}\exp \left\{ \frac{\Im ^{2}z}{24t}\beta
^{-1}-Q\beta ^{-\frac{\varepsilon }{1-\varepsilon }}\right\} .
\label{eq5.4.0}
\end{equation}
To prove Theorem \ref{thm1.2} we only need to consider the case $n=1$. We
have 
\begin{multline}  \label{eq5.4.1}
\int_{0}^{\varkappa /2}\frac{d\beta }{\beta }\int_{-\infty }^{\infty
}\left\vert F\left( \alpha ,\beta \right) \right\vert d\alpha \\
\lesssim _{z,t,q_{+}}\int_{0}^{\varkappa /2}\beta ^{-3/2}\exp \left\{ \frac{%
\Im ^{2}z}{24t}\beta ^{-1}-Q\beta ^{-\frac{\varepsilon }{1-\varepsilon }%
}\right\} d\beta .
\end{multline}
So $F$ is subject to the condition of Lemma \ref{lem5.1} if the integral in %
\eqref{eq5.4.1} converges, which depends on $\varepsilon$ and $\Im z$.

\begin{itemize}
\item[\protect\underline{Case 1}.] $\varepsilon >1/2$. Then\footnote{%
We assume $\varepsilon <1$. If $\varepsilon \geq 1$ then Theorem \ref{thm1.2}
is trivial.} $\frac{\varepsilon }{1-\varepsilon }>1$ and the right hand side
of \eqref{eq5.4.1} is finite for any $z\in \mathbb{C}$. This means that $%
\mathbb{M}_{+}(z,t)$ is an entire $\mathfrak{S}_{1}$-valued function for any 
$t>0$ and due to \eqref{eq5.2.1}, we can pass to the limit in \eqref{eq5.1}
as $n\rightarrow \infty $ by standard properties of infinite determinants
(see e.g. \cite{GGK00}). This proves \eqref{it1} in Theorem \ref{thm1.2}.

\item[\protect\underline{Case 2}.] $\varepsilon =1/2$. Then $\frac{%
\varepsilon }{1-\varepsilon }=1$ and the right hand side of \eqref{eq5.4.1}
converges if and only if 
\begin{equation*}
\frac{\Im ^{2}z}{24t}-Q<0
\end{equation*}%
or when 
\begin{equation*}
\left\vert \Im z\right\vert <\sqrt{12Q}\cdot \sqrt{t}.
\end{equation*}%
Choosing the maximum possible value of $Q$ in \eqref{eq7.1} we get 
\begin{equation*}
\left\vert \Im z\right\vert <\frac{9\sqrt{2}}{8}c\sqrt{t}
\end{equation*}%
and \eqref{it2} of Theorem \ref{thm1.2} follows. Thus, Theorem \ref{thm1.2}
is proven.

\item[\protect\underline{Case 3}.] $0<\varepsilon <1/2$. Then $\frac{%
\varepsilon }{1-\varepsilon }<1$ and \eqref{eq5.4.1} clearly diverges for
any $\Im z\neq 0$ and our method fails to establish analyticity and we have
to go back to \eqref{eq5.4.0} and analyze it for any natural $n$. Expanding $%
\left( \left\vert \alpha \right\vert +\beta \right) ^{n-1}$ in %
\eqref{eq5.4.0} by the binomial formula we have 
\begin{multline}  \label{eq5.6.1}
\int_{0}^{\varkappa /2}\frac{d\beta }{\beta }\int_{-\infty }^{\infty
}F(\alpha ,\beta )d\alpha \\
\lesssim _{z,t,q_{+}}\sum_{k=0}^{n-1}\binom{n-1}{k}\int_{0}^{\varkappa
/2}d\beta \beta ^{k-1}e^{-Q\beta ^{-\frac{\varepsilon }{1-\varepsilon }%
}}\int_{0}^{\infty }\alpha ^{n-k-1}e^{-24\beta t\alpha ^{2}}d\alpha .
\end{multline}%
Reducing the inner integral in \eqref{eq5.6.1} to the gamma function%
\footnote{%
Recall $\Gamma (z)=\displaystyle\int_{0}^{\infty }\alpha ^{z-1}e^{-\alpha
}d\alpha $.}, 
\begin{align}
\eqref{eq5.6.1}& =\sum_{k=0}^{n-1}\binom{n-1}{k}\int_{0}^{\varkappa
/2}d\beta \beta ^{k-1}e^{-Q\beta ^{-\frac{\varepsilon }{1-\varepsilon }%
}}\cdot \frac{1}{(3\beta t)^{\frac{n-k}{2}}}\Gamma \left( \frac{n-k}{2}%
\right)  \notag \\
& \lesssim \sum_{k=0}^{n}\binom{n}{k}(3t)^{-\frac{n-k}{2}}\Gamma \left( 
\frac{n-k}{2}\right) \int_{0}^{1}d\beta \beta ^{\frac{3k-n}{2}-1}e^{-Q\beta
^{-\frac{\varepsilon }{1-\varepsilon }}}.  \label{eq5.7.1}
\end{align}%
Introducing in the last integral the new variable $s=\beta ^{-\frac{%
\varepsilon }{1-\varepsilon }}$ and setting $\gamma \equiv\frac{%
1-\varepsilon }{2\varepsilon }>1/2$ we get 
\begin{align}
\int_{0}^{1}d\beta \beta ^{\frac{3k-n}{2}-1}e^{-Q\beta ^{-\frac{\varepsilon 
}{1-\varepsilon }}}& =\frac{1-\varepsilon }{\varepsilon }\int_{1}^{\infty
}s^{\frac{\varepsilon -1}{\varepsilon }\left( \frac{3k-n}{2}-1\right)
-1}e^{-Qs}ds  \notag \\
& \lesssim \int_{1}^{\infty }s^{\gamma (n-3k)-1}e^{-Qs}ds  \notag \\
& \lesssim Q^{-\gamma (n-3k)+1}\int_{Q}^{\infty }s^{\gamma (n-3k)}e^{-s}ds.
\label{eq5.7.2}
\end{align}%
The behavior of the last integral depends on the sign of $\omega
_{k}\equiv\gamma (n-3k)$. If $\omega _{k}\geq 0$, i.e. $3k\leq n$, then 
\begin{align*}
J_{k}& \equiv \int_{Q}^{\infty }s^{\omega _{k}}e^{-s}ds \\
& \leq \int_{0}^{\infty }s^{\omega _{k}-1}e^{-s}ds=\Gamma (\omega _{k}) \\
& =\Gamma (\gamma (n-3k)).
\end{align*}%
If $\omega _{k}<0$, i.e. $3k>n$, then 
\begin{equation*}
J_{k}\leq Q^{\omega _{k}-1}\int_{Q}^{\infty }e^{-s}ds\leq Q^{\omega _{k}-1}.
\end{equation*}
Splitting the sum in \eqref{eq5.7.1} accordingly, we see that the right hand
side of \eqref{eq5.7.1} is dominated by 
\begin{multline}  \label{eq5.8.1}
\sum_{0\leq 3k\leq n}\binom{n}{k}(3t)^{-\frac{n-k}{2}}Q^{-\omega _{k}}\Gamma
\left( \frac{n-k}{2}\right) \Gamma \left( \omega _{k}\right) +\sum_{n<3k\leq
3n}\binom{n}{k}(3t)^{-\frac{n-k}{2}}\Gamma \left( \frac{n-k}{2}\right) \\
\equiv S_{1}+S_{2}.
\end{multline}

Analyze now $S_{1}$ and $S_{2}$. For $S_{1}$ we have 
\begin{align}
S_{1}& \leq \Gamma \left( \frac{n}{2}\right) \Gamma (\gamma n)\sum_{0\leq
3k\leq n}\binom{n}{k}(3t)^{-\frac{n-k}{2}}Q^{-\omega _{k}}  \notag \\
& \leq \left( Q^{2\gamma }+\frac{1}{\sqrt{3t}Q^{\gamma }}\right) ^{n}\Gamma
\left( \frac{n}{2}\right) \Gamma (\gamma n).  \label{eq5.9.1}
\end{align}%
For $S_{2}$ we obtain 
\begin{equation*}
S_{2}\leq \left( 1+\frac{1}{\sqrt{3t}}\right) ^{n}\Gamma \left( \frac{n}{3}%
\right)
\end{equation*}%
and hence the contribution from $S_{2}$ to \eqref{eq5.7.2} produces a real
analytic function. On the other hand, as it easily follows from %
\eqref{eq5.9.1}, the contribution from $S_{1}$ produces a function from $%
G^{\gamma -1/2}=G^{\frac{1}{2\varepsilon }-1}$. Thus we have proven that if $%
0<\varepsilon <1/2$ then 
\begin{equation*}
\mathbb{M}_{+}(x,t)=\mathbb{M}_{+}^{\left( 1\right) }(x,t)+\mathbb{M}%
_{+}^{\left( 2\right) }(x,t)
\end{equation*}%
where $\mathbb{M}_{+}^{\left( 1\right) }(x,t)$ is a real analytic $\mathfrak{%
S}_{1}$-valued function and $\mathbb{M}_{+}^{\left( 2\right) }(x,t)$ is a $%
\mathfrak{S}_{1}$-valued function from the Gevrey class $G^{\frac{1}{%
2\varepsilon }-1}$. Thus, we can pass to the limit as before. The limiting
function is from the Gevrey class $G^{\frac{1}{2\varepsilon }-1}$ if $\det
(1+\mathbb{M}(x,t))$ doesn't vanish for all $x\in \mathbb{R}$. The latter
occurs if $1+\mathbb{M}(x,t)$ is invertible on $\mathbb{R}$ for any $t>0$.
\end{itemize}

Theorem \ref{thm1.3} is proven.

\section{Discussions, corollaries, and open problems}

\subsection{Discussions}


\begin{remark}
Theorem \ref{thm1.2} improves our main result from \cite{Ry11} where $%
\mathbb{M}\left( x,t\right) \in \mathfrak{S}_{1}$ was not proven and only
real analyticity of $u\left( x,t\right) $ was obtained. The main idea of 
\cite{Ry11} is to put together the analytic continuation arguments of \cite%
{Ryb10} to treat initial data on $\mathbb{R}_{-}$ and Tarama's approach from 
\cite{Tarama04} to handle the data on $\mathbb{R}_{+}$. As far as we know
the solution to Problem \ref{pb1} given in \cite{Tarama04} was best known
back then. The main result of \cite{Tarama04} says that $u\left( x,t\right) $
is real analytic under the following conditions: $q$ is real and $L_{%
\limfunc{loc}}^{2}$ such that 
\begin{equation*}
\int_{-\infty }^{\infty }\left( 1+\left\vert x\right\vert \right) \left\vert
q\left( x\right) \right\vert dx<\infty
\end{equation*}%
and for $x$ large enough there are positive $C_{q},c$ so that%
\begin{equation*}
\int_{x}^{\infty }\left\vert q\right\vert ^{2}\leq C_{q}e^{-cx^{1/2}}.
\end{equation*}%
Note that these conditions are much stronger than Hypothesis \ref{hyp1.1}.
The techniques used in \cite{Tarama04} are also based upon the (classical)
IST but his analysis relies on the properties of the Airy function as
opposed to ours which is based on analytic and pseudo-analyitc
continuations. The latter appears particularly well-suited for addressing
Problem \ref{pb1} and consequently significantly less involved.\bigskip
\end{remark}


\begin{remark}
It is proven in \cite{Deift79}, Theorem 7.2 that if $q$ is analytic in the
strip $\left\vert \func{Im}z\right\vert <a$ and has Schwartz decay there,
then $u\left( z,t\right) $ is meromorphic in a strip with at most $N$ poles
(where $N$ is the number of bound states of $\ H_{q}$) off the real line. By
Theorem 7.1 from the same \cite{Deift79}, for the reflection coefficient one
then have $R\left( \lambda \right) =O\left( e^{-2a\left\vert \lambda
\right\vert }\right) $ as $\lambda \rightarrow \infty $ which of course need
not occur in our case. This implies that our real meromorphic solution $%
u\left( z,t\right) $ in Theorem \ref{thm1.2} has, in general, infinitely
many poles for any $t>0$ in any strip around the real line accumulating only
to infinity. By general theorems \cite{Steinberg69} on families of compact
meromorphic operators these poles continuously depend on $t$ and hence may
appear or disappear only on the boundary of analyticity of $u\left(
x,t\right) $ (including infinity).
\end{remark}

\subsection{Corollaries}

The following statement is a direct consequence of the analyticity of $%
u\left( z,t\right) $ for $t>0$.

\begin{corollary}
\label{Corollary'}Under conditions of Theorem \ref{thm1.2} the solution $%
u\left( z,t\right) $ can not vanish on an open set for any $t>0$ unless $q$
is identically zero.
\end{corollary}

This quickly recovers and improves a number of unique continuation results
due to Zhang \cite{Zhang92}. E.g., one of the main results of \cite{Zhang92}
says that $u\left( x,t\right) $ cannot have compact support at two different
moments unless it vanishes identically. The techniques of \cite{Zhang92}
rely upon the classical IST (coupled with some Hardy space arguments) and
are valid under certain decay and regularity conditions on $q$.

\begin{corollary}
\label{Corol2} The class of (nonsmooth) initial data $q$ such that 
\begin{equation}
\int_{-\infty }^{\infty }e^{c\left\vert x\right\vert ^{\varepsilon
}}\left\vert q\left( x\right) \right\vert dx<\infty \text{ for some }%
c,\varepsilon >0\text{ }  \label{exp_decay}
\end{equation}%
is not preserved under the KdV flow. 
\end{corollary}

\begin{proof}
Assume that for some $t=t_{0}$ the function $u\left( x,t_{0}\right) $ is
subject to (\ref{exp_decay}). Since the KdV equation is invariant under the
transformation $\left( x,t\right) \rightarrow \left( -x,-t\right) $, the
solution $u_{0}\left( x,t\right) $ to the problem (\ref{eq1.1}) with the
initial data $q_{0}\left( x\right) =u\left( -x,t_{0}\right) $, by Theorems %
\ref{thm1.2}, \ref{thm1.3}, will be at least smooth for any $t>0$. But $%
u_{0}\left( x,t_{0}\right) =q\left( x\right) $ forcing original $q$ to be
smooth too.
\end{proof}

Corollary \ref{Corol2}, in turn, implies that under the KdV flow neither an
exponential decay at $-\infty $ nor smoothness persist in general. Note in
this connection that issues related to persistence of regularity are also
very important and have been extensively studied but we don't touch on this
here.

The explicit formula (\ref{det_form}), which was used to derive our
analyticity results, does have some practical value. E.g. it implies that
the large time asymptotic behavior of $u\left( x,t\right) $ is completely
determined by the measure $\rho (\lambda )$ in (\ref{eq4.1}) alone. This
fact is so far rigorously proven for $q$'s tending to a negative constant or
a periodic function at $-\infty $ and was used to obtain explicit
expressions for the so-called asymptotic solitons (see, e.g. \cite{Hruslov76}%
, \cite{Ven86}, and \cite{KhrKot94}). We plan to return to this important
issue elsewhere.

\subsection{Open problems}

\begin{enumerate}
\item We believe that under Hypothesis \ref{hyp1.1} our solutions $u\left(
x,t\right) $ have no singularities on the real line for any $t>0$. If this
held then the problem \eqref{eq1.1} would be globally well-posed under
Hypothesis \ref{hyp1.1} only and no blow-up solution could develop. That is
to say that $1+\mathbb{M}(x,t)$ is automatically invertible for any real $x$
and $t>0$ under Hypothesis \ref{hyp1.1} alone. This fact is quite easy if in
(\ref{eq4.1}) the support of $\rho \left( \lambda \right) $ is rich enough
(a set of uniqueness of an analytic function) or $\left\vert R\left( \lambda
\right) \right\vert <1$ on any set of positive Lebesgue measure (see \cite%
{Ry11}). The situation is much less trivial if $R\left( \lambda \right) $ in
(\ref{eq4.1}) is unimodular for a.e. real $\lambda $ (i.e. $q$ is completely
reflecting). An affirmative answer is given in \cite{GR11} for the case of $%
q $ such that $q|_{\mathbb{R}_{+}}=0$ and $H_{q}\geq 0$ (absence of negative
spectrum). To address the problem as stated one needs to show that $1+%
\mathbb{M}(x,t)$ is invertible in the case when in (\ref{eq4.1}) $\rho
\left( \lambda \right) $ is supported on a set $\left\{ \lambda _{n}\right\}
\subset \mathbb{R}_{+}$ such that $\dsum \lambda _{n}<\infty $ and $%
\left\vert R\left( \lambda \right) \right\vert =1$ a.e. on the real line. In
term of the Schrodinger operator $H_{q}$ itself this means that the
absolutely continuous spectrum of $H_{q}$ is simple and supported on $%
\mathbb{R}_{+}$ but there is a rich embedded positive singular spectrum.
Physically relevant examples can be constructed from the Gaussian white
noise, Pearson sparse blocks, Kotani potentials, etc.

\item We do not know much about the Banach (or Hilbert) space of meromorphic
function to which $u\left( z,t\right) $ from \ref{thm1.2} belongs. It would
be very interesting to find such spaces as this would give, among others,
important norm estimates for $u\left( z,t\right) $ which our paper lacks.

\item We (cautiously) conjecture that in Theorem \ref{thm1.3} $u\left(
x,t\right) $ could be represented for any $t>0$ as a meromorphic function
plus a small Gevrey regular function. We can in fact show that the trace
norm of $\mathbb{M}^{\left( 2\right) }(x,t)$ from Theorem \ref{thm1.3} can
be made small but it is not clear if after taking the $\det $ and then $\log 
$ the analytic and small Gevrey parts will still be separated. Of course,
this question will immediately have an affirmative answer if under
conditions of Theorem \ref{thm1.3} the solution $u\left( x,t\right) $
happens to be real analytic. Our methods however fail to yield such results.
\end{enumerate}

\section*{Acknowledgement}

We are grateful to Fritz Gesztesy for valuable discussions.

\bibliographystyle{abbrv}
\bibliography{bibryb}

\end{document}